\documentclass[journal,10 pt]{IEEEtran} 
\usepackage{epsfig,color,amsmath,cite}
\usepackage{amsthm} 
\usepackage{amsmath}    
\usepackage[T1]{fontenc}
\usepackage[utf8]{inputenc}
\usepackage{bm}
\usepackage{epstopdf}
\usepackage{amssymb}
\usepackage{url}
\usepackage{enumitem} 
\usepackage{multirow}
\usepackage{hhline}
\usepackage{booktabs}
\usepackage{mathtools}
\usepackage{makecell}
\usepackage[linesnumbered,boxed,commentsnumbered,ruled,vlined,longend]{algorithm2e}
\usepackage{comment}

\DeclareMathOperator*{\minimize}{minimize}
\DeclareMathOperator*{\maximize}{maximize}

\DeclareMathOperator*{\subjectto}{subject\ to}

\makeatother
\DeclareMathAlphabet\mathbfcal{OMS}{cmsy}{b}{n}

\newtheorem{mydef}{Definition}
\newtheorem{mylem}{Lemma}

\newtheorem{myprs}{Proposition}



\usepackage{stackengine}

\newcommand{\mat}[1]{\boldsymbol{#1}}

\newcommand{\bmat}[1]{\begin{bmatrix} #1 \end{bmatrix}}

\providecommand{\mA}{\ensuremath{\mat{A}}}

\newcommand{\m}{\boldsymbol}
\allowdisplaybreaks[4]
\usepackage[colorlinks = true,
linkcolor = blue,
urlcolor  = blue,
citecolor = blue,
anchorcolor = blue]{hyperref}

\newcommand{\mbb}[1]{\mathbb{#1}}
\newcommand{\mr}[1]{\mathrm{#1}}
\usepackage[framemethod=TikZ]{mdframed}
\mdfdefinestyle{MyFrame}{%
	linecolor=black,
	outerlinewidth=1.25pt,
	roundcorner=1.25pt,
	innerrightmargin=5pt,
	innerleftmargin=5pt,}
	
\usepackage[noabbrev]{cleveref}

\usepackage{mathtools}

\DeclarePairedDelimiter\abs{\lvert}{\rvert}%
\DeclarePairedDelimiter\norm{\lVert}{\rVert}%

\makeatletter
\let\oldabs\abs
\def\abs{\@ifstar{\oldabs}{\oldabs*}}
\let\oldnorm\norm
\def\norm{\@ifstar{\oldnorm}{\oldnorm*}}
\makeatother

\usepackage[english]{babel}
\usepackage[utf8]{inputenc}
\usepackage[super]{nth}

\usepackage{graphicx}
\usepackage{float}
\usepackage[caption = false]{subfig}

\usepackage{array}
\usepackage{threeparttable}

\usepackage[english]{babel}
\usepackage[utf8]{inputenc}
\usepackage[super]{nth}

\SetKwRepeat{Do}{do}{while}%

\usepackage{lettrine} 
\IEEEoverridecommandlockouts

\title{\Large \vspace{0.7cm} \LARGE \centering {\textsc{\textbf{State-Robust Observability Measures for Sensor Selection in Nonlinear Dynamic Systems}}}}

\author{Mohamad H. Kazm$\text{a}^{1}$, Sebastian A. Nugroh$\text{o}^{2}$, Aleksandar Habe$\text{r}^{3}$, and Ahmad F. Tah$\text{a}^{1}$\vspace{-0.5cm}
	\thanks{
	$^{1}$Department of Civil and Environmental Engineering, Vanderbilt University, 2201 West End Ave, Nashville, Tennessee 37235, USA 
	(mohamad.h.kazma@vanderbilt.edu, ahmad.taha@vanderbilt.edu.)}
	\thanks{
	$^{2}$Cummins Technical Center, Cummins Inc., Columbus, IN 47201 USA  (sebastian.nugroho@cummins.com)}
	\thanks{$^3$Department  of  Manufacturing and Mechanical Engineering Technology,  College of Engineering Technology, Rochester Institute of Technology, 1 Lomb Memorial Dr, Rochester, NY 14623, USA (ml.mecheng@gmail.com).}
	\thanks{This work is partially supported by National Science Foundation (NSF) under Grants 2152450 and 2151571.}
}

\begin{document}

\setlength{\abovedisplayskip}{3.1pt}
\setlength{\belowdisplayskip}{3.1pt}
\setlength{\abovedisplayshortskip}{3.1pt}
\setlength{\belowdisplayshortskip}{3.1pt}
		
\newdimen\origiwspc%
\newdimen\origiwstr%
\origiwspc=\fontdimen2\font
\origiwstr=\fontdimen3\font


\maketitle

\markboth{To Appear in the 62$^{\text{nd}}$ IEEE Conference on Decision and Control (CDC'2023), Singapore, Decemeber 2023}{}

\newcommand\reduline{%
	\bgroup\markoverwith
	{\textcolor{red}{\pgfsetfillopacity{0.2}\rule[-0.5ex]{2pt}{10pt}\pgfsetfillopacity{1}}%
		\textcolor{red}{\llap{\rule[0.4ex]{2pt}{0.4pt}}\llap{\rule[0.7ex]{2pt}{0.4pt}}}%
	}%
	\ULon}

\begin{abstract}
This paper explores the problem of selecting sensor nodes for a general class of nonlinear dynamical networks. In particular, we study the problem by utilizing altered definitions of observability and open-loop lifted observers. The approach is performed by discretizing the system's dynamics using the implicit Runge-Kutta method and by introducing a state-averaged observability measure. The  observability measure is computed for a number of perturbed initial states in the vicinity of the system's true initial state. The sensor node selection problem is revealed to retain the submodular and modular properties of the original problem. This allows the problem to be solved efficiently using a greedy algorithm with a guaranteed performance bound while showing an augmented robustness to unknown or uncertain initial conditions. The validity of this approach is numerically demonstrated on a $H_{2}/O_{2}$ combustion reaction network.
\end{abstract}


\vspace{-0.2cm}
\section{Introduction and Paper Contributions}
\lettrine[lines=2]{T}{he} sensor selection problem is one of the fundamental control engineering problems. The problem is crucial for the control, monitoring, and safe operation of a large number of engineered systems, such as electric power grids \cite{Taylor2017}, municipal water networks \cite{Berry2005}, and transportation systems \cite{Mehr2018}. From a control- and observability-based formulation, this problem aims to find the optimal combination of sensor nodes (graph nodes whose local states should be observed) that optimize appropriate observability measures. 
The goal is to make the system \textit{as observable as possible} using a limited number of sensors to be placed on select nodes in the network. 

Sensor selection problems have gained considerable research interest in recent years as a plethora of methods have been proposed in the literature, especially for \textit{linear} systems. These methods can be categorized based on underlying mathematical approaches, such as network and graph theory \cite{ZHANG2017202,Pequito2016}, sparsity promoting algorithms \cite{Dhingra2014,Argha2019}, semidefinite approximations and relaxations \cite{Taha2019timevarying}, heuristic optimization under convex relaxations~\cite{Joshi2009}, greedy approach under submodular set maximization~\cite{Summers2016Submodularity}, and mixed-integer optimization \cite{Taylor2017,Nugroho2019}. Regardless, methods for solving sensor selection problems for nonlinear dynamic networks are significantly less developed. Only a handful of methods have been proposed so far to address this problem for nonlinear dynamic systems. 

A sensor selection algorithm for target tracking in nonlinear dynamic networks based on a generalized information gain is proposed in \cite{Shen2014sensor}. Next, an empirical observability Gramian approach is utilized in \cite{Qi2015empirical} for placing phasor measurement units in transmission power networks. Another approach based on an open-loop moving horizon estimation for sensor selection and state observation is proposed in \cite{Haber2017}. The approach presented in~\cite{Haber2017} is more numerically tractable than the approaches based on empirical observability Gramians. A new randomized algorithm is presented in \cite{Bopardikar2019} in which theoretical bounds for eigenvalues and condition numbers of observability Gramians are developed. A novel framework is proposed in \cite{Nugroho2019Sensor} for sensor selection and observer design. This approach is developed by using the Lyapunov stability theory and mixed-integer semidefinite optimization. Lastly, methods to place actuators in nonlinear networks that are based on heuristically solving mixed-integer nonlinear optimization problems have been recently developed in \cite{Haber2021}. 

Here it should be emphasized that most of the developed approaches for solving the sensor selection problems, especially the ones involving mixed-integer programs, are \textit{not} necessarily efficient and scalable for large nonlinear dynamic networks. The computational burden of the developed approaches becomes significant even for small or medium-sized nonlinear networks. Another issue with sensor selection problems for nonlinear networks is that,  in practice, the initial states of the system are usually not known \textit{a priori}. This creates model uncertainties and difficulties in formulating and solving the sensor selection problem since the numerically tractable observability-based approach~\cite{Haber2017} for nonlinear systems involves a dependency on initial states. This implies that under such state-dependency any perturbation to the initial state tends to yield in most cases different sensor node selections for the same nonlinear network modeled under similar system parameters.

To partly address the aforementioned limitations, we extend observability-based sensor selection method introduced in~\cite{Haber2017} by introducing state-averaged observability measures for nonlinear networks. That is, instead of utilizing the observability measures associated with a single guess on the initial state, we consider a state-averaged observability metric that relies on several points located around the actual initial state. This allows the constructed observability-based measures to take into account the variabilities resulting from initial conditions perturbations on the sensor selection measures.

Note that due to the structure preserving operations that yield the state-averaged observability Gramian, we show that the observability metrics for sensor selection retain modularity and submodularity properties. This consequently allows the sensor selection problem to be solved through greedy heuristic, and thus make it suitable to solve sensor selection for  large-scale nonlinear dynamic networks. 

Accordingly, the main contributions of this paper are:
\begin{itemize}[leftmargin=*]
	\item We introduce a \textit{state-averaged observability measure} for sensor selection in nonlinear networks. We use a number of points located around the nonlinear system's initial state. By relying on such local state-averaged observability measure, we attain an optimal sensor selection that is robust against unknown or uncertain initial conditions.
	\item We provide theoretical and numerical validation that under such state-averaged observability measures the  submodularity and modularity of the sensor selection objective function is retained. In particular, we leverage the modularity and submodularity of the $\mr{trace}$ and $\mr{log}$-$\mr{det}$ measures of the constructed observability Gramian to perform the sensor selection. Under such formulation, greedy algorithms are employed to solve the combinatorial set optimization problem and as a consequence the selection problem is rendered scalable to large-scale nonlinear dynamic networks.
	\item  We evaluate the validity and robustness of the proposed approach by providing descriptive numerical experiments that showcase the proposed sensor selection strategy. The method is tested on a nonlinear $H_{2}/O_{2}$ combustion reaction network. 
\end{itemize}

This rest of the paper is organized as follows. Section \ref{sec:problem_formulation} introduces the problem formulation. Section \ref{sec:submmodular_greedy} presents some theoretical results pertaining to the state-averaged observability measures. Numerical results are presented in Section \ref{sec:numerical_tests}, and Section \ref{sec:conclusion} concludes the paper.

\vspace{0.2cm}

\noindent {\textit{Paper's Notation:}}~Let $\mathbb{R}$, $\mathbb{R}^n$, and $\mathbb{R}^{p\times q}$ denote the set of real numbers, and real-valued row vectors with size of $n$, and $p$-by-$q$ real matrices. The cardinality of the a set $\mathcal{N}$ is denoted by $|\mathcal{N}|$. The symbol $\otimes$ denotes the Kronecker product. The identity matrix of size $n$ is denoted by $\m{I}_n$.
The operators $\mr{log}$-$\mr{det}(\mA)$ returns the logarithmic-determinant of matrix $\m{A}$, $\mr{trace}(\m{A})$ returns the trace of matrix of matrix $\m{A}$. The operator $\mathrm{diag}\{{a_{i}}\}_{{i}=1}^{\mr{n}} \in \mathbb{R}^{n\times n}$ constructs a block diagonal matrix with scalar $a_i$ as the diagonal entries for all $i\in \{1, \dots, \mr{N}\}$. The operator $\mathrm{col}\{\m{x}_{i}\}_{i=0}^{\mr{N}} \in \mathbb{R}^{\mr{N}.n}$ constructs a column vector that concatenates vectors $\m{x}_i \in \mathbb{R}^{n}$ for all $i \in \{0, \dots, \mr{N}\}$. For any vector $\m x \in \mathbb{R}^{n}$, $\Vert\m x\Vert_2$ denotes the Euclidean norm of $\m x$, defined as $\Vert \m x\Vert_2 := \sqrt{\m x^{\top}\m x} $ , where $\m x^{\top}$ is the transpose of $\m x$.

\section{Preliminaries and Problem Formulation}\label{sec:problem_formulation}
In this section, we introduce mathematical preliminaries and define the problem of selecting sensor nodes. We consider a general nonlinear dynamic network defined in~\eqref{eq:model_CT} under a  continuous-time representation.
\begin{subequations}\label{eq:model_CT} 
	\begin{align}
\dot{\m x}(t) &= \m f(\m x(t)), \label{eq:model_CT_1} \\
	\m y(t) &=  \m{\Gamma} \m{C}{\m x}(t),\label{eq:model_CT_2} 
	\end{align}
\end{subequations}
where $\m x\in\mathbb{R}^{n_x}$ is the global state and $\m y\in \mathbb{R}^{n_{y}}$ is the global output vector. The nonlinear mapping function $\m f:\mbb{R}^{n_x}\rightarrow\mbb{R}^{n_x}$ is assumed to be smooth and at least twice continuously differentiable.
The measurement matrix $\m C\in \mbb{R}^{n_y\times n_x}$ is assumed to be known. The matrix $\m\Gamma := \mathrm{diag}\{{\gamma_j}\}_{{j}=1}^{n_y}\hspace{-0.05cm}\in \mbb{R}^{n_y\times n_y}$ determines the configuration of the sensors---that is, a node $j$ is equipped with a sensor if $\gamma_j= 1$. Otherwise, we simply set $\gamma_j= 0$. We define the parameterize vector $\m \gamma$ that represents the sensor selection, i.e, $\m \gamma \hspace{-0.05cm}=\hspace{-0.05cm}\mathrm{col}\{\gamma_j\}_{j=1}^{n_y}\hspace{-0.05cm}$. Without the loss of generality, we have assumed that the inputs are not affecting the system dynamics.

The objective of the sensor selection problem for the nonlinear dynamics \eqref{eq:model_CT} is to determine the combination of sensors (the $1$ and $0$ patterns in $\m \gamma$) such that an observability-based metric is maximized under a sensor ratio constraint. As such, in order to formulate the binary selection problem, we refer to utilizing a discrete-time representation of the nonlinear state model \eqref{eq:model_CT_1}. 


There exists several methods that can be utilized to obtain a discrete-time model. The choice of discretization method must rely upon 
the system's \textit{stiffness}, desired accuracy, and the performance of computation resources. 
 In this paper, we consider the discretization of \eqref{eq:model_CT} using the \textit{implicit Runge-Kutta} (IRK) method \cite{iserles2008first}. The main advantage of IRK method is that it can be applied to a wider class of nonlinear dynamic networks with various degree of stiffness. Readers can refer to~\cite{Atkinson2011} for the discrete-time modeling techniques of nonlinear systems. The methodology herein results in the following implicit discrete-time state-space model
\begin{align}
	\begin{split}
	\m \zeta_{1,k+1} &= \m x_{k}+ \tfrac{T}{4}\left(\m f(\m \zeta_{1,k+1})-\m f(\m \zeta_{2,k+1})\right), \\
	\m \zeta_{2,k+1} &= \m x_{k}+ \tfrac{T}{12}\left(3\m f(\m \zeta_{1,k+1})+5\m f(\m \zeta_{2,k+1})\right), \\
	\m x_{k+1} &= \m x_{k}+\tfrac{T}{4}\left(\m f(\m \zeta_{1,k+1})+3\m f(\m \zeta_{2,k+1})\right),
	\end{split}\label{eq:TI_dynamics}
\end{align}
where $T > 0$ denotes the discretization period, $k\in\mbb{N}$ is the discrete-time index such that $\m x_k = \m x(kT)$,
and $\m \zeta_{1,k+1},\m \zeta_{2,k+1}\in\mbb{R}^{n_x}$ are auxiliary vectors for computing $\m x_{k+1}$ provided that $\m x_{k}$ is given. Notice that in order to compute $\m x_{k+1}$, we first need to solve a system of nonlinear equations that consists of the first two equations in \eqref{eq:TI_dynamics}. The unknowns in this system are $\m \zeta_{1,k+1},\m \zeta_{2,k+1}$. This layer of complexity is necessary since the introduced discrete-time model can accurately and in a numerically stable manner represent a broad class of nonlinear networks, including networks with stiff dynamics. 
\vspace{-0.5cm}
\subsection{Initial  State Estimation}
\vspace{-0.15cm}
Taking into account the model~\eqref{eq:TI_dynamics}, the discrete-time equivalent of nonlinear dynamic network \eqref{eq:model_CT} can be compactly written in the following form
\begin{subequations}
\begin{align}
	\m x_{k+1} &= \m x_{k} +\tilde{\m f} (\m x_{k+1}, \m x_k), \label{eq:TI_dynamics_compact} \\
	\m y_k &= \m \Gamma \m C \m x_k, \label{eq:DT_measurement_model}
\end{align}
\end{subequations}
where the function $\tilde{\m f}(\cdot)$ in \eqref{eq:TI_dynamics_compact} represents the implicit dynamics in \eqref{eq:TI_dynamics}.  
The proposed approach for sensor selection is developed using the concept of an open-loop lifted observer framework. To that end, we introduce the lifted vector  $\tilde{\m y}\in \mbb{R}^{\mr{N}.n_y}$ that is constructed as $\tilde{\m y} = \mathrm{col}\{\tilde{\m y}_{i}\}_{i=1}^{\mr{N}-1} $. The positive integer $\mr{N}$ is the observation window. For the sake of simplicity, it is assumed temporarily that $\m \Gamma$ is fixed such that the output measurement equation \eqref{eq:DT_measurement_model} is reduced to $ \tilde{ \m y}_k = \tilde{\m C} \m x_k$, where $\tilde{\m C}$ is obtained by compressing the zero rows of $\m\Gamma \m C$. Now, define the vector function $\m h:\mbb{\m R}^{n_x}\rightarrow\mbb{\m R}^{\mr{N}.n_y}$ as 

\begin{equation}~ \label{eq:initial_state_est_fun}
	\underbrace{\bmat{\m h_{0}(\m x_0) \\ \m h_{1}(\m x_0) \\ \vdots \\ \m h_{\mr{N}-1}(\m x_0)}}_{\m h (\m x_0)} := \underbrace{\bmat{\tilde{\m y}_{0} \\ \tilde{\m y}_{1} \\ \vdots \\ \tilde{\m y}_{\mr{N}-1}}}_{\tilde{\m y}} - \underbrace{\bmat{\m g_{0}(\m x_0) \\ \m g_{1}(\m x_0) \\ \vdots \\ 		\m g_{\mr{N}-1}(\m x_0)}}_{\m g (\m x_0)},
\end{equation}
where $\m h(\m x_0) = \mathrm{col}\{\m h_{i}(\m x_0)\}_{i=0}^{\mr{N}-1}$. The function  $\m g:\mbb{\m R}^{n_x}\rightarrow\mbb{\m R}^{\mr{N}.n_y}$ is constructed as $\m g(\m x_0) = \mathrm{col}\{\m g_{i}(\m x_0)\}_{i=0}^{\mr{N}-1}$, where $\m g_{i}:\mbb{\m R}^{n_x}\rightarrow\mbb{\m R}^{n_y}$ and $\m g_{i}:= \tilde{\m C}\m x_{i}$ for all $i\in \{0,1,2,\cdots,\mr{N}-1\}$. It is understood from \eqref{eq:TI_dynamics_compact} that $\m g_{i}$ is a function of only of the initial state $\m x_0$ due to the fact that $\m x_i$ is a recursive function of $\m x_0$ for each $i$.
Consequently, we can write
\begin{align}
	\m h(\m x_0) = 0 \Leftrightarrow \tilde{\m y}  = \m g(\m x_0). \label{eq:initial_state_zero} 
\end{align}

Since in practice the actual initial state is unknown \textit{a priori}, then for a fixed selection of sensors,  $\m x_0$ can be estimated by solving the following nonlinear state estimation optimization problem with a predefined weighting matrix $\m Q \succ 0$ such that
\begin{subequations}
\begin{align}
\hspace{-0.3cm}	{(\mathbf{P1})}\;\;\,\,	 \minimize_{\hat{\m x}_0\in \mathbfcal{X}}\,\,\,\, &{\m h(\hat{\m x}_0)}^\top \m Q {\m h(\hat{\m x}_0)}\\
	\subjectto\,\,\,\, &{\hat{\m x}}^l_0 \leq \hat{\m x}_0 \leq {\hat{\m x}}^u_0,
\end{align}
\end{subequations}
where ${\hat{\m x}}^l_0$ and ${\hat{\m x}}^u_0$ are respectively the lower and upper bounds of $\hat{\m x}_0$ and $\m Q \in \mathbb{R}^{\mr{N}.n_{y}\times \mr{N}.n_{y}}$ is the weighting matrix. The weight matrix  $\m Q$ assigns weights to the measured states such that $\mathbf{P1}$, the initial state estimation problem, efficiently reaches a solution.
\subsection{Observability-based Sensor Node Selection}
Observability of nonlinear discrete-time systems can be quantified using the concept of \textit{uniform observability}~\cite{Hanba2009}. That is, the system \eqref{eq:TI_dynamics_compact} with the measurement model \eqref{eq:DT_measurement_model} is said to be uniformly observable in $\mathbfcal{X}$ ($\mathbfcal{X}$ is the subset representing a local operating region of \eqref{eq:TI_dynamics_compact}) if there exists a finite $\mr{N} \in \mbb{N}$ such that the relation $\tilde{\m y} = \m g\left(\m x_0\right)$ is injective (one-to-one) with respect to $\m x_0\in \mathbfcal{X}$ for any given set of measured outputs $\tilde{\m y}$.


Accordingly, if $\m g\left(\cdot\right)$ is injective with respect to $\m x_0$, then $\m x_0$ can be uniquely determined from the set of measurements $\tilde{\m y}$.

As such, let $\m {J}_g(\cdot)$ be a Jacobian matrix of the function $\m g(\cdot)$ around  $\m x_0$.  
A \textit{sufficient} condition for the mapping $\m g(\cdot)$ to be injective is that the Jacobian matrix of $\m g(\cdot)$ is of full rank~\cite{Hanba2009}.

The Jacobian matrix $\m {J}_g(\cdot)\in\mbb{R}^{\mr{N}.n_y\times n_x}$ is given as 
\begin{align}
	\m {J}_g(\m \gamma,\m x_0) :=  \dfrac{d\m g(\m x_0)}{d\m x_0} 
 =\hspace{-0.05cm}\mathrm{col}\left\{\dfrac{\partial \m g_i(\m x_0)}{\partial \m x_0}\right\}_{{i}=0}^{\mr{N}-1}\hspace{-0.05cm}. \label{eq:obs_jacobian_x0}
\end{align}

For each $i\in \{0,1,2,\cdots,\mr{N}-1\}$, the term $\tfrac{\partial \m g_i(\m x_0)}{\partial \m x_0}$ in \eqref{eq:obs_jacobian_x0} is equivalent to
\begin{align}
	\dfrac{\partial \m g_i(\m x_0)}{\partial \m x_0} = \dfrac{\partial }{\partial \m x_0} \tilde{\m C} \m x_i = \tilde{\m C}  \dfrac{\partial \m x_i}{\partial \m x_0} = \tilde{\m C}\prod_{j=0}^{i-1} \dfrac{\partial \m x_{j+1}}{\partial \m x_j}. \label{eq:partial_g_i}
\end{align}

It is important to mention that the computation of $\tfrac{\partial \m x_i}{\partial \m x_0}$ in \eqref{eq:partial_g_i} requires the knowledge of $\m x_j$ for all $j$. The value of $\m x_j$ 
can be obtained by simulating \eqref{eq:TI_dynamics_compact} with the initial condition $\m x_0$. 
Taking into account the parameterized measurement equation \eqref{eq:DT_measurement_model}, the Jacobian matrix $\m {J}_g(\cdot)$ in \eqref{eq:obs_jacobian_x0} around a specific initial state $\hat{\m x}_0$ is given as
\begin{align}
		\m {J}_g(\m \gamma,\hat{\m x}_0) := \m {J}_g(\hat{\m x}_0) = \bmat{\m I \otimes \m \Gamma \m C}\times {\m \xi(\hat{\m x}_0)},\label{eq:obs_jacobian_x0_param}
\end{align}

where $\m \xi : \mbb{R}^{n_x}\rightarrow\mbb{R}^{\mr{N}.n_x}$, $\m \xi(\hat{\m x}_0) = \mathrm{col}\{\m \xi_{i}(\hat{\m x}_0)\}_{i=0}^{\mr{N}-1}$, and $\m \xi_i := \tfrac{\partial \hat{\m x}_i}{\partial \hat{\m x}_0}$.
Next, we define the matrix function $ \m W(\cdot):\mbb{R}^{n_x}\rightarrow\mbb{R}^{n_x\times n_x}$ as the following
\begin{align}
	\m W(\m \gamma, \m x_0) := \m {J}_g^{\top}(\m x_0)\m {J}_g(\m x_0).\label{eq:obs_gramian_x0}
\end{align} 

The matrix $\m W(\cdot)$ is fundamental for the analysis and solving the system of nonlinear equations as well as for the development of methods presented in this paper.  Namely, the spectral properties of the matrix $\m W(\cdot)$ determine the convergence properties of the Newton's method used for solving the system of nonlinear equations~\eqref{eq:initial_state_zero}~ \cite{nocedal2006numerical}. Note that, in a general case, this matrix is \textit{not} equal to the observability Gramian for linear systems, since constructing it involves the computation of partial derivatives of the IRK equations \eqref{eq:TI_dynamics}. We note here that we have referred to the use of implicit IRK method since it accounts for a wide class of nonlinear networks, however other implicit discretization methods can be utilized to formulate the observability-based sensor selection problem. In our previous work~\cite{Kazma2023}, we perform optimal sensor selection for a class of differential algebraic equations under the trapezoidal implicit method~\cite{Atkinson2011} discretization. 

Motivated by the fact that this matrix is closely related 
 to the empirical observability Gramian~\cite{Qi2015empirical,Haber2017}, we will refer to this matrix as the observability Gramian of the discrete-time system~\eqref{eq:TI_dynamics_compact}-\eqref{eq:DT_measurement_model}. Notice that the Gramian matrix \eqref{eq:obs_gramian_x0} contains the matrix $\tilde{\m C}$, which is a function of the vector $\m \gamma$.

To that end, the sensor selection problem can be mathematically formulated as follows. Let $X=\{\hat{\m x}_{0}^{(1)},\hat{\m x}_{0}^{(2)},\ldots, \hat{\m x}_{0}^{(q)}\}$ be a set of initial conditions of the dynamics~\eqref{eq:TI_dynamics_compact}-\eqref{eq:DT_measurement_model}. This set of initial conditions is chosen by the user. Furthermore, let $r$ be the final number of sensor nodes that is also specified by the user. Then, the sensor nodes are selected as the solution of the following integer optimization problem
\begin{subequations}\label{eq:ssp_gramian}
	\begin{align}
		(\mathbf{P2})\;\;\maximize_{\m\gamma} \;\;\;
		&\mathcal{O}\left(\m \gamma, X \right) \label{eq:ssp_gramian_1}\\
		\subjectto \;\;\;\,& \;\sum_{i=1}^{n_y} \gamma_i = r, \; \m \gamma\in\{0,1\}^{n_y}, \label{eq:ssp_gramian_2}
	\end{align}
\end{subequations}
where $\mathcal{O}\left(\m \gamma, X \right)$ is a user-selected function that quantifies the observability of the system.

The main idea of our approach is to incorporate a number of initial conditions into the function $\mathcal{O}(\cdot)$ that quantifies the system observability. This is because the ``exact'' initial condition of the system is usually uncertain. By relying on a state-averaged observability matrix that is constructed under several predictions/perturbations of the initial state, the sensor selection procedure becomes less sensitive to uncertainties on initial states that are necessary to solve the system of nonlinear equations.

One approach for tackling sensor selection problems within networks, is posing such  combinatorial problem as a constraint set maximization problem~\cite{Krause2008, Summers2016Submodularity}. The rationality behind such approach is later evident when solving to the sensor selection problem, where underlying set function properties allow for a scalable solution to the optimization problem.
As such, the sensor node selection problem $\mathbf{P2}$ can be rewritten as a set maximization problem $\mathbf{P3}$ by defining 
the \textit{set function} $\mathcal{O}{(\mathcal{S})}: 2^{\mathcal{V}}\rightarrow \mbb{R}$ with $\mathcal{V} := \{ i\in\mbb{N}\,|\,0 < i \leq n_y\}$. Herein, the set $\mathcal{V}$ denotes the set of all possible combinations of sensor locations.
\begin{align}
	{(\mathbf{P3})}\;\;	\mathcal{O}^*_{\mathcal{S}} := \maximize_{\mathcal{S}\subseteq\mathcal{V}}\,\, f(\mathcal{S}),\;\; \subjectto\,\, \abs{\mathcal{S}} = r.
\end{align}

 In the context of sensor selection, $\mathbf{P3}$ translates to the problem of finding the best sensor configuration $\mathcal{S}$ containing $r$ number of sensors such that a particular observability metric is maximized. The variable $\m{\Gamma}$ is encoded in the set $\mathcal{S}$, such that for each sensor node a value of $\gamma_j$ is attributed to the set $\mathcal{S}$ at location $j$.

\vspace{-0.05cm}
\section{Observability-Based Sensor Selection}\label{sec:submmodular_greedy}
\vspace{-0.05cm}
In this section, we introduce several observability measures, quantify their properties, and present our approach for solving the problem~\eqref{eq:ssp_gramian}. Our approach is based on defining a state-averaged observability measure and using a greedy algorithm to efficiently solve the sensor selection problem. The justification of using the greedy algorithm will be established by showing that the introduced set function measures retain set function properties---modularity or submodularity.
For the development of our approach we need to obtain a closed-form expression for \eqref{eq:obs_gramian_x0}. The following proposition establishes this expression.
\begin{myprs}\label{prs:param_observ_gramian}
	The parametrized observability Gramian~\eqref{eq:obs_gramian_x0} for the nonlinear discrete-time dynamic networks \eqref{eq:TI_dynamics_compact} with parametrized measurement model \eqref{eq:DT_measurement_model} around a particular initial state $\hat{\m x}_0$ can be expressed as follows
	\begin{align}
		\m W(\m \gamma,\hat{\m x}_0) = \sum_{j=1}^{n_y} \gamma_j \left(\sum_{i=0}^{\mr{N}-1}\left(\dfrac{\partial \hat{\m x}_i}{\partial \hat{\m x}_0}\right)^{\hspace{-0.1cm}\top} \hspace{-0.075cm}\m c_j^\top \m c_j \dfrac{\partial \hat{\m x}_i}{\partial \hat{\m x}_0}\right),\label{eq:param_obs_gramian}
	\end{align}	
	where $\m c_j\in\mbb{R}^{1\times n_x}$ is the $j$-th row of $\m C$. 
\end{myprs} 
\begin{proof}
	From  \eqref{eq:obs_jacobian_x0_param} and \eqref{eq:obs_gramian_x0}, it follows that
	\begin{align}
		\m W(\m \gamma,\hat{\m x}_0) &= \m \xi^\top(\hat{\m x}_0)\bmat{\m I \otimes \m \Gamma \m C}^\top \bmat{\m I \otimes \m \Gamma \m C}\m \xi(\hat{\m x}_0) \nonumber \\
		&= \sum_{i=0}^{M-1}\m \xi_i^\top(\hat{\m x}_0) \m C^\top \m \Gamma^2 \m C \m \xi_i(\hat{\m x}_0)\nonumber \\
		&= \sum_{i=0}^{M-1}\m \xi_i^\top(\hat{\m x}_0) \left(\sum_{j=1}^{n_y} \gamma_j\m c_j^\top \m c_j \right) \m \xi_i(\hat{\m x}_0)\nonumber \\
		&= \sum_{i=0}^{M-1}\sum_{j=1}^{n_y}\gamma_j\m \xi_i^\top(\hat{\m x}_0)  \m c_j^\top \m c_j \m \xi_i(\hat{\m x}_0), \label{eq:param_obs_gramian_proof}
	\end{align}
	which holds since $\m \Gamma^2 = \m \Gamma$. Since \eqref{eq:param_obs_gramian_proof} is equivalent to \eqref{eq:param_obs_gramian}, then the proof is complete.
\end{proof}

\subsection{Modular \& Submodular Set Functions}\label{ssec:mod_subod}


There exist several observability measures and metrics. Usually, these measures are expressed on the basis of the rank, smallest eigenvalue, condition number, trace, and determinant of an appropriate matrix---see \cite{Qi2015empirical} and the references therein. Such measures have set function properties, modularity and submodularity, that allow greedy algorithm to solve the optimal sensor selection problem. 
The following definition characterizes modular and submodular set functions \cite{Summers2016Submodularity,lovasz1983submodular}.

\begin{mydef} 
	\label{def:modular_submodular}
	A set function $	\mathcal{O}: 2^{\mathcal{V}}\rightarrow \mbb{R}$ is said to be modular if and only if for any $\mathcal{S}\subseteq\mathcal{V}$ and weight function $w:\mathcal{V}\rightarrow \mbb{R}$ it holds that
	\begin{subequations}
		\begin{align}
			\mathcal{O}(\mathcal{S}) = w(\emptyset) + \sum_{s\in\mathcal{S}} w(s), \label{eq:modular_def}
		\end{align}
		and $	\mathcal{O}(\cdot)$ is said to be submodular if and only if for any $\mathcal{A},\mathcal{B}\subseteq\mathcal{V}$ given that $\mathcal{A}\subseteq\mathcal{B}$, it holds that for all $s\notin\mathcal{B}$
		\begin{align}
			\mathcal{O}(\mathcal{A}\cup\{s\}) - 	\mathcal{O}(\mathcal{A})\geq 	\mathcal{O}(\mathcal{B}\cup\{s\}) - 	\mathcal{O}(\mathcal{B}). \label{eq:submodular_def}
		\end{align}
	\end{subequations}
\end{mydef}

As seen from \eqref{eq:submodular_def}, for any submodular function, the addition of an element $s$ to a smaller subset $\mathcal{A}$ yields a greater reward compared to adding the same element to a bigger subset $\mathcal{B}$. This notion is normally termed as \textit{diminishing return property} \cite{Summers2016Submodularity}. Aside from modularity and submodularity, the notion of monotone increasing and decreasing functions are also important to achieve scalable sensor selection.

\begin{mydef} 
	\label{def:monotone_increasing}
	A set function $	\mathcal{O}: 2^{\mathcal{V}}\rightarrow \mbb{R}$ is called monotone increasing if, for $\mathcal{A},\mathcal{B}\subseteq\mathcal{V}$, $\mathcal{A}\subseteq\mathcal{B}$ implies $	\mathcal{O}(\mathcal{B})\geq 	\mathcal{O}(\mathcal{A})$ and called monotone decreasing if $\mathcal{A}\subseteq\mathcal{B}$ implies $	\mathcal{O}(\mathcal{A})\geq 	\mathcal{O}(\mathcal{B})$.
\end{mydef}

In retrospect with the sensor selection problem posed in $\mathbf{P3}$, the parametrized observability Gramian associated with $\mathcal{S}\subseteq\mathcal{V}$ around a presumed initial state $\hat{\m x}_0$ is defined as 
\begin{align}
	\tilde{\m W}(\mathcal{S},\hat{\m x}_0) := \sum_{j\in \mathcal{S}} \left(\sum_{i=0}^{\mr{N}-1}\m\left(\xi_i\right)^\top \hspace{-0.05cm}\m c_j^\top \m c_j \xi_i\right).~\label{eq:param_obs_gramian_set}
\end{align}

It is worthwhile to note that the notation $j\in \mathcal{S}$ corresponds to every activated sensor such that $\gamma_j = 1$. If the chosen form of the observability measure function  renders $\mathbf{P3}$ to be submodular and monotone increasing, then the greedy algorithm can be used to efficiently determine sensor locations. The greedy algorithm is summarized in Algorithm \ref{algorithm1}. If the function $	\mathcal{O}(\cdot)$ is submodular and monotone increasing, and if the set of sensor locations computed using the greedy algorithm is $\mathcal{S}$, then we have the following performance guarantee\cite{nemhauser1978analysis}
\begin{align*}
	\frac{	\mathcal{O}^*-	\mathcal{O}(\mathcal{S})}{	\mathcal{O}^*-	\mathcal{O}(\emptyset)} \leq \left(\frac{r-1}{r}\right)^r \leq \frac{1}{e},\;\;
\end{align*}
where $\mathcal{O}^*$ is the optimal value of $\mathbf{P3}$ and $e\approx 2.71828$. Note that the above worst-case bound is merely theoretical. For submodular set maximization it has been shown that an accuracy of $99\%$ is achieved~\cite{Summers2016Submodularity}. 
\begin{algorithm}[h]
	\caption{Greedy Algorithm \cite{Summers2016Submodularity}}\label{algorithm1}
	\DontPrintSemicolon
	\textbf{input:} $r$, $\mathcal{V}$ \;
	\textbf{initialize:} $\mathcal{S}\leftarrow \emptyset$, $k \leftarrow 1$ \;
	\While{$k \leq r$}{
		\textbf{compute:} $\mathcal{G}_k = \mathcal{O}(\mathcal{S}\cup \{a\})-\mathcal{O}(\mathcal{S})$, $\forall a\in \mathcal{V}\setminus \mathcal{S}$ \;
		\textbf{assign:} $\mathcal{S}\leftarrow \mathcal{S} \cup \left\{\mathrm{arg\,max}_{a\in \mathcal{V}\setminus \mathcal{S}}\,\mathcal{G}_k \right\}$ \;
		\textbf{update:} $k \leftarrow k + 1$\;
	}
	\textbf{output:} $\mathcal{S}$\;
\end{algorithm}
\vspace{-0.3cm}
\subsection{State-Averaged Observability Sensor Selection}
Ideally, the parametrized Gramian \eqref{eq:param_obs_gramian} should be constructed using the system's actual initial state ${\m x}_0$. Nonetheless,  this state is usually unknown or only some vector entries are known \textit{a priori}.
In practice, we only have a guess of the initial state, that is denoted by $\hat{\m x}_0$. To minimize the variability from quantifying the observability around $\hat{\m x}_0$, we opt to use a state-averaged observability metric which, instead of computing the observability Gramian around a single guess of initial state $\hat{\m x}_0$, alternatively it is computed by taking into account several points of presumed initial states $\hat{\m x}_0^{(\kappa)}$ for $\kappa \in \{1,2,\cdots,q\}$. Using this concept of state-averaged observability,  we introduce the following metric
\begin{align}
\mathcal{O}\left(\mathcal{S}, X \right)=\frac{1}{q}\sum_{\kappa=1}^{q}\mathcal{L}\left(\tilde{\m W}^{(\kappa)}(\mathcal{S},\hat{\m x}_0^{(\kappa)}) \right),
\label{metric1}
\end{align}
where $\mathcal{L}(\cdot)$ is an appropriate function mapping matrix into a scalar
\begin{align}
	\tilde{\m W}^{(\kappa)}(\mathcal{S},\hat{\m x}_0^{(\kappa)}) := \sum_{j\in \mathcal{S}} \left(\sum_{i=0}^{M-1}\m\left(\xi_i^{(\kappa)}\right)^\top \hspace{-0.05cm}\m c_j^\top \m c_j \xi_i^{(\kappa)}\right).
\end{align}

This form of the Gramian matrix is established on the basis of Proposition~\ref{prs:param_observ_gramian}. Using this new measure, 
$\mathbf{P3}$ is developed further into the following set optimization problem
\begin{subequations}\label{eq:ssp_gramian_mean}
	\begin{align}
		(\mathbf{P4})\;\;\maximize_{\mathcal{S}} \;\;\;
		&\mathcal{O}(\mathcal{S}):=\frac{1}{q}\sum_{\kappa=1}^{q}\mathcal{L}\left(\tilde{\m W}^{(\kappa)}(\mathcal{S}) \right),\label{eq:ssp_gramian_mean_1}\\
		\subjectto \;\;\;\,& \; \abs{\mathcal{S}} = r, \; \mathcal{S}\subseteq\mathcal{V}.  \label{eq:ssp_gramian_mean_2}
	\end{align}
\end{subequations} 
\subsection{Modularity \& Submodularity of the Proposed Measures}\label{ssec:mod_subod_obs_metrics}
In the sequel we will analyze the modularity and submodularity properties of the average observability metrics \eqref{metric1}. We will analyze the cases when the function $\mathcal{L}(\cdot)$ is $\mr{trace}$, and $\mr{log}$-$\mr{det}$. 
The following Lemma provides support to the analysis on  
the modularity, submodularity, and monotonicity properties of the average observability metric $\mathcal{O}(\cdot)$, when the function $\mathcal{L}(\cdot)$ in $\mathbf{P4}$ takes the form of the $\mr{trace}$ and $\mr{log}$-$\mr{det}$.

\begin{mylem}\label{lem:convex_comb}
	For set functions $\mathcal{L}_{1}, \mathcal{L}_{2}, \dots, \mathcal{L}_{k}: 2^{\mathcal{V}}\rightarrow \mbb{R}$ that are submodular. 
	Any conic combination, that is, any weighted non-negative sum defined as 
 	\begin{equation}
 		\mathcal{O}(\mathcal{S})\hspace{-0.05cm}:=\hspace{-0.05cm}\sum_{\kappa=1}^{q}w_{k}\mathcal{L}_{k},
 	\end{equation}
	is submodular, such that $w_{k}\geq 0 \;\forall \;k$.
\end{mylem}
\begin{proof}\label{proof:conic}
		We prove the submodularity of a non-negative weighted sum from the definition of submodularity. As such, from Def.~\ref{def:modular_submodular}, we have $\mathcal{A},\mathcal{B}\subseteq\mathcal{V}$ given that $\mathcal{A}\subseteq\mathcal{B}$, and that for all $s\notin\mathcal{B}$
	\begin{align*}
	&\mathcal{L}_{k}(\mathcal{A}\cup\{s\}) - 	\mathcal{L}_{k}(\mathcal{A})\geq 	\mathcal{L}_{k}(\mathcal{B}\cup\{s\}) - 	\mathcal{L}_{k}(\mathcal{B}),
	\end{align*}
	then under a conic combination and based on Def.~\ref{def:modular_submodular} the following holds true
	\begin{align*}
	&\sum_{\kappa=1}^{q}w_{k} \Big(\mathcal{L}_{k}(\mathcal{A}\cup\{s\}) -
	\mathcal{L}_{k}(\mathcal{A})\Big) \\
	&\geq \sum_{\kappa=1}^{q}w_{k} \Big(\mathcal{L}_{k}(\mathcal{B}\cup\{s\}) - \mathcal{L}_{k}(\mathcal{B})\Big),
	\end{align*}
	for any $\mathcal{A},\mathcal{B}\subseteq\mathcal{V}$ given that $\mathcal{A}\subseteq\mathcal{B}$, and for all $s\notin\mathcal{B}.$
\end{proof}

Conic combinations along with set restrictions and contractions are submodularity preserving operations~\cite{Bach2013b}. 
Lemma~\ref{lem:convex_comb} shows that submodularity of the original submodular functions is retained under a non-negative weighted sum and thus formulates the rationale behind developing a state-averaged observability metric. As such, the following proposition shows that the state-averaged $\mr{trace}(.)$ metric is modular. 
\begin{myprs}\label{prs:trace_prop}
	A set function			
	 $\mathcal{O}:\hspace{-0.05cm}2^{\mathcal{V}}\hspace{-0.05cm}\rightarrow\hspace{-0.05cm}\mbb{R}$ defined by 
	 \begin{align}
	\mathcal{O}(\mathcal{S})\hspace{-0.05cm}:=\hspace{-0.05cm}\frac{1}{q}\sum_{\kappa=1}^{q}\mathrm{trace}\left(\tilde{\m W}^{(\kappa)}(\mathcal{S})\right), \label{eq:trace_mod} 
	\end{align}
	for $\mathcal{S}\subseteq\mathcal{V}$ is modular.
\end{myprs}
\begin{proof}
For any $\mathcal{S}\subseteq\mathcal{V}$, observe that
\begin{align*}
&\frac{1}{q}\sum_{\kappa=1}^{q}\mathrm{trace}\left(\tilde{\m W}^{(\kappa)}(\mathcal{S})\right)\\
	&\quad = \frac{1}{q}\sum_{\kappa=1}^{q}\mathrm{trace}\left(\sum_{j\in \mathcal{S}} \left(\sum_{i=0}^{M-1}\m\left(\xi_i^{(\kappa)}\right)^\top \hspace{-0.05cm}\m c_j^\top \m c_j \xi_i^{(\kappa)}\right)\hspace{-0.05cm}\right) \\
	&\quad = \sum_{j\in \mathcal{S}}  \left(\frac{1}{q}\sum_{\kappa=1}^{q}\mathrm{trace}\left(\sum_{i=0}^{M-1}\m\left(\xi_i^{(\kappa)}\right)^\top \hspace{-0.05cm}\m c_j^\top \m c_j \xi_i^{(\kappa)}\right)\hspace{-0.05cm}\right),
\end{align*}	
thus showing that $\mr{trace}(\cdot)$ is a linear matrix function and therefore is modular. 
\end{proof}
\begin{figure}[t]
	\centering
	\vspace{-0.1cm}
	\hspace{-0.1cm}
	\subfloat{\includegraphics[keepaspectratio=true,scale=0.55]{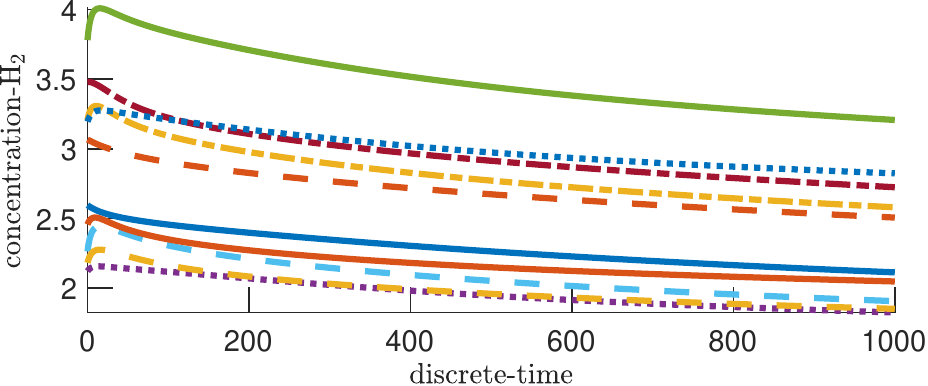}}\vspace*{-0.2cm} \hspace{-0.22cm}
	\subfloat{\includegraphics[keepaspectratio=true,scale=0.55]{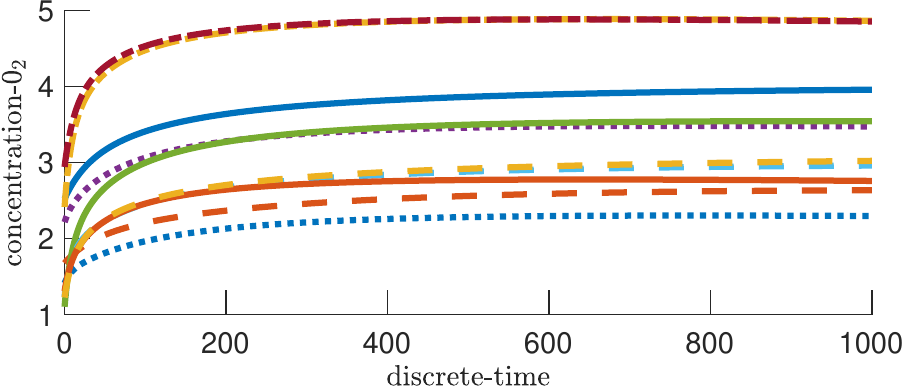}}{}{}
	\caption{State trajectories of the simulated $H_{2}/O_{2}$ combustion network under perturbed initial conditions. The states are concentrations of $H_{2}$ and $O_{2}$ chemical species. }\vspace{-0.6cm}
	\label{fig1new}
\end{figure}

The state-averaged $\mr{log}$-$\mr{det}(\cdot)$ observability metric is submodular and monotone increasing.
\begin{myprs}\label{prs:trace_det_rank_prop}
A set function $\mathcal{O}:\hspace{-0.05cm}2^{\mathcal{V}}\hspace{-0.05cm}\rightarrow\hspace{-0.05cm}\mbb{R}$ characterized by 	
\begin{align}
	\mathcal{O}(\mathcal{S})\hspace{-0.05cm}:=\hspace{-0.05cm}\frac{1}{q}\sum_{\kappa=1}^{q}\mathrm{log\,det}\left(\tilde{\m W}^{(\kappa)}(\mathcal{S})\right), \label{eq:logdet_submodular} 
\end{align} 
for $\mathcal{S}\subseteq\mathcal{V}$ is submodular and monotone increasing.
\end{myprs}
\begin{proof} For brevity we do not provide the full proof regarding the submodularity and the increasing monotonicity of the $\mr{log}$-$\mr{det}(\cdot)$. Such metric is well studied in the field of submodular optimization and is proved therein---readers are referred to~\cite{Krause2011, Summers2016Submodularity, Zhou2019, Bilmes2022}.
For the state-averaged observability metric $\mathcal{O}(\mathcal{S})$ in~\eqref{eq:logdet_submodular} and based on Lemma~\ref{lem:convex_comb}, the  submodularity of the set function $\mr{log}$-$\mr{det}(\cdot)$ under a non-negative weighted sum is preserved and thus it is submodular.
\end{proof}
The following section showcases the robustness of the sensor selection problem that is based on the proposed state-averaged observability metrics and that it is solved via scalable greedy heuristics.

\section{Numerical Studies}\label{sec:numerical_tests}
In this section, we numerically validate and investigate the effectiveness of the averaged-observability based sensor selection framework. To numerically test our methods, we choose a general nonlinear model of a combustion reaction network. Consider the following list of $N_{r}$ chemical reactions
\begin{align}
	\sum_{i=1}^{n_{x}}q_{ji} \mathcal{R}_{i} \rightleftarrows  
	 \sum_{i=1}^{n_{x}} w_{ji}\mathcal{R}_{i}, \; j=1,2,\ldots, N_{r},
	\label{chemicalReactionNetwork}
\end{align}
where $q_{ji}$ and $w_{ji}$ are stoichiometric coefficients and  $\mathcal{R}_{i}$, $i\in \{1,2,\cdots,n_{x}\}$, are chemical species (notice that the number of chemical species is equal to the global state dimension). With the chemical reactions described in \eqref{chemicalReactionNetwork}, we associate a state-space model. In this representation, the states are concentrations of chemical species. The resulting state equation has the following form~\cite{turns1996introduction,smirnov2014modeling}
\begin{align}
	\dot{\m {x}}(t)=\Theta \boldsymbol{\psi}\left(\m {x}(t)\right),
	\label{continiousTimeFinal}
\end{align}
where $\boldsymbol{\psi}\left(\m {x}\right)=[\psi_{1}\left(\m {x} \right),\psi_{2}\left(\m {x} \right),\ldots, \psi_{n_{r}}\left(\m {x} \right)]^{T}$, and $\Theta=[w_{ji}-q_{ji} ]\in \mathbb{R}^{n_{x}\times N_{r}}$, and  
$\m {x}=[x_{1},x_{2},\ldots, x_{n_{x}}]$, where $x_{i}$, $i\in \{1,2,\cdots,n_{x}\}$ are concentrations of chemical species, and finally $\psi_{j}$, $j=1,2,\ldots, N_{r}$ are the polynomial functions of concentrations defined as follows
\begin{align}
	\psi_{j}\left(\m {x}\right)= v_{j}\prod_{i=1}^{n_{x}}x_{i}^{q_{ji}}-b_{j}\prod_{i=1}^{n_{x}}x_{i}^{w_{kj}},\;r=1,2,\ldots, N_{r},
	\label{polynomialQ}
\end{align}
where $v_{j},b_{j}\in \mathbb{R}_{+}$ are the forward and backward reaction rates that are computed on the basis of the Arrhenius law. 
\begin{figure}[t]
	\vspace{-0.4cm}
	\hspace{-0.1cm}
	\subfloat{\includegraphics[keepaspectratio=true,scale=0.67]{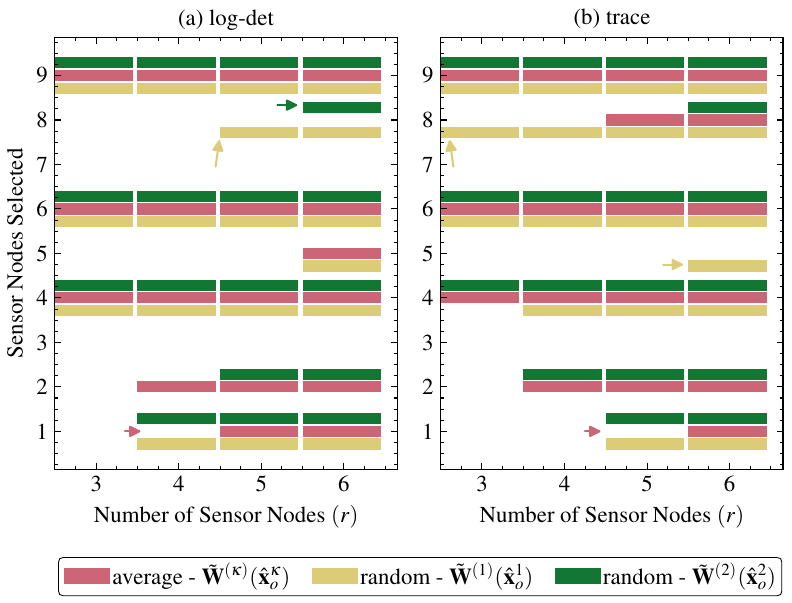}}\vspace*{-0.35cm} \hspace{-0.25cm}
	\caption{Selected sensor nodes resulting from state-averaged observability measures and observability measures that are based on a single randomly selected initial condition. Arrows represent the changes relative to state-averaged metrics.}\vspace{-0.6cm}
	\label{fig:SNS}
\end{figure}

 


In this paper, we consider an $H_{2}/O_{2}$ combustion network. This network has $27$ reactions and $9$ chemical species. The reaction rates are computed using the Cantera software~\cite{Cantera}. We use a chemical reaction network model described in the Cantera database file ``h2o2.cti". In our computations, we assume a temperature of $2500\;[K]$ and an initial pressure equal to the atmospheric pressure. We have chosen a smaller combustion network in order to be able to compare our methods with randomized sensor node placements. To discretize the dynamics we use a discretization constant of $T=1\cdot 10^{-12}$ and we assume the observation window of $\mr{N} = 1000$. We have chosen such value of the discretization constant by analyzing an initial condition response of the system.

\begin{figure}[t]
	\centering 
	\vspace{+0.1cm}
	\subfloat{\includegraphics[keepaspectratio=true,scale=0.5]{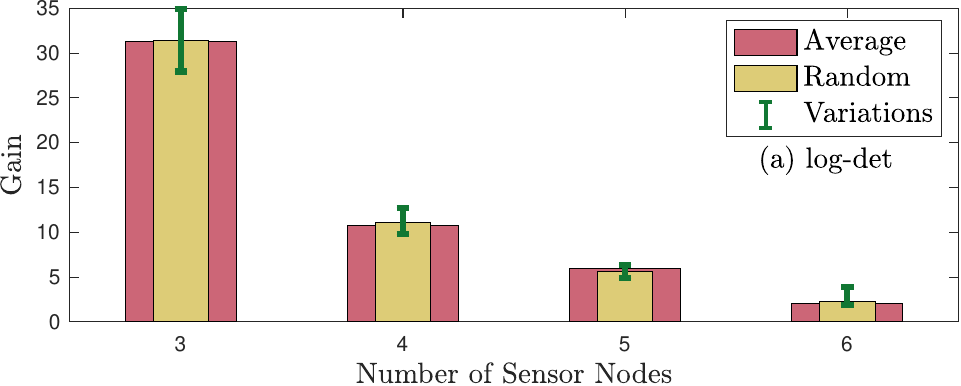}}\vspace*{-0.15cm}{}
	\hspace{-1cm} 
	\subfloat{\includegraphics[keepaspectratio=true,scale=0.5]{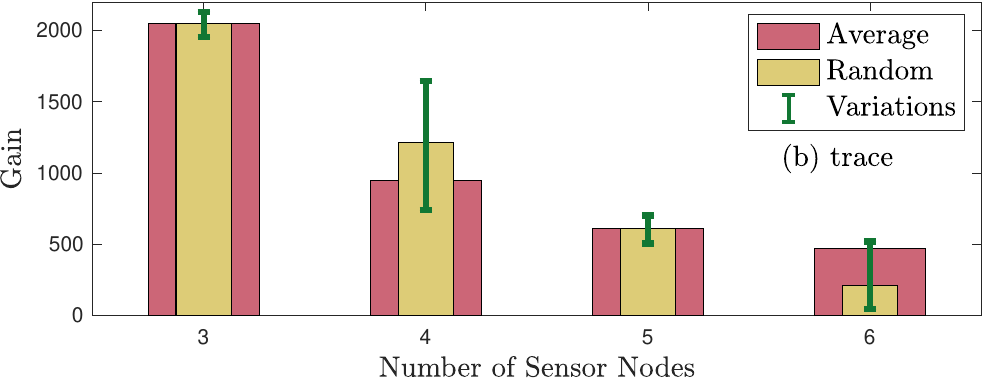}}{}
	\caption{Gain on observability measures with additional sensor node selections. The average (red) represents that of state-averaged observability and that of perturbed single guesses is represented by an average gain (yellow) and the max and min variations on that gain (green).} \vspace{-0.3cm}
	\label{fig:gain}
\end{figure}

We assume $\kappa=10$ and states are selected as random perturbations of the ``true" state $\m {x}=[2,\; 0, \; 0, \; 1, \; 0, \; 0, \; 0, \; 0.2, \; 0]^{T}$. The random perturbations are drawn from a uniform perturbation on the interval $[0,p]$. Fig.~\ref{fig1new} depicts the state trajectories of $H_{2}$ and $O_{2}$ from the simulated $H_{2}/O_{2}$ combustion network under uniform perturbation interval with $p=2$. We note that for each state, different state trajectories are obtained when starting with different initial condition. That is, for each initial condition, the trajectory of the system tends to a different attractor. In dynamical systems, a basin of attractor is a state condition that the systems tends towards as it evolves over a time period~\cite{Dudkowski2016}. This suggests that the nonlinear system~\eqref{continiousTimeFinal} has several basins of attraction and as such, we investigate how such perturbed state trajectories affect the sensor node selection model and asses the robustness of the proposed state-averaged observability measures.
\subsection{Robust Observability-based Sensor Node Selection}
\begin{figure*}[t]
	\centering
	\hspace{-0.12cm}
	\subfloat{\includegraphics[keepaspectratio=true,scale=1]{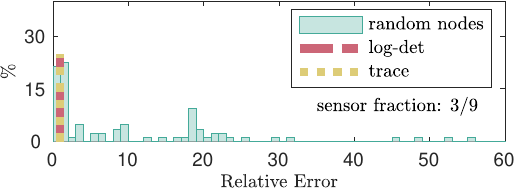}}
	\subfloat{\includegraphics[keepaspectratio=true,scale=1]{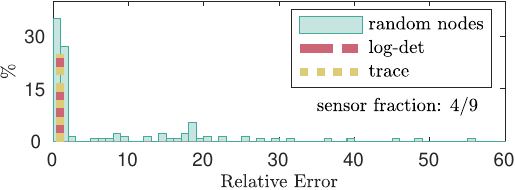}}{}{}
	\subfloat{\includegraphics[keepaspectratio=true,scale=1]{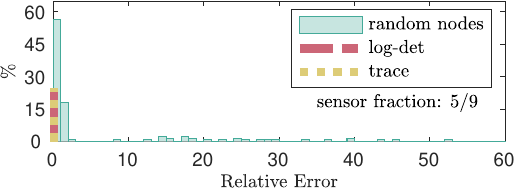}}
	\subfloat{\includegraphics[keepaspectratio=true,scale=1]{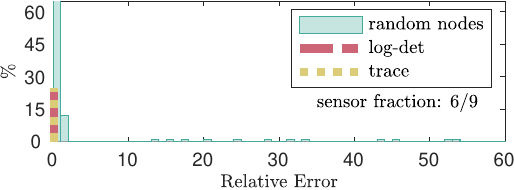}}{}{}
	\caption{Relative error on initial state estimation. Histograms represent the error from randomly selected nodes with specified sensor fractions. The y-axis represents the percentage of sensor configurations that resulted in a certain relative error in the initial state estimation (histograms). The lines, having a percentage of $100\%$, represent the error under the optimal solutions to problem $\mathbf{P4}$.}~\label{fig:relativeerror}
\end{figure*}

Our first goal is to determine optimal sensor locations using the greedy algorithm and the proposed state-averaged observability measures. These observability measures are computed on the basis of the perturbed states  (the ``true" state is not used to compute the observability measure). 
Fig.~\ref{fig:SNS} represents the sensed node locations determined by solving $\mathbf{P4}$ based on the state-averaged observability metrics $(\mr{log}$-$\mr{det},\mr{trace})$ and on the observability metrics associated with a single guess of the perturbed initial state. It can be pointed out that starting with different initial guesses, the optimal set of selected nodes is also different. The different selection as compared to that of the averaged metric are pointed by arrows. We note here that to evaluate the robustness of the state-averaged metric, we perform $\mathbf{P4}$ based on a random generation of $\kappa=10$ initial guesses, such that the same solution is obtained for each generated initial guesses. Thus, showing that the proposed metrics are robust to choice and perturbations of the initial conditions.

Understanding the underlying theory that allows us to optimally choose of $\kappa$ for a specified perturbation $p$ and ensure robustness on $\mathcal{S}^{*}$ from initial state perturbations is outside the scope of this paper and will be investigated in future work. 
For now we note that for the presented general nonlinear network a choice of $\kappa=10$, that is greater than the number of sensed node $n_p =9$, results in robust optimal sensor node selections for a range of $p \leq 20$, which is a relatively high perturbation given the true state. As such, one can infer that the perturbation magnitude is not the critical factor for the choice of $\kappa$. The authors suggest that this is related to the number of states/sensor nodes and the stability of the state trajectories resulting from the perturbed states.

To further understand the performance of the state-averaged observability measure as compared with that based on a single initial guess, we investigate the variations on observability gain; Fig.~\ref{fig:gain} depicts the gain on the observability metrics resulting from an additional sensor selection relative to the prior number of sensed nodes. This variability in gain value for each of the $\mr{log}$-$\mr{det}$ and $\mr{trace}$ metrics based on single initial conditions represents where the change in sensed nodes occurs. For instance, under states with different gain value a sensed node might be forfeited for another. It can be depicted that when the average of the gain resulting from a single perturbed initial conditions is different than that of the state-average metric, a different sensed node is chosen. For example considering the $\mr{log}$-$\mr{det}$ metric and when the number of nodes the be chosen is increasing from 4 to 5, we notice that this average differs. Referring back to Fig.~\ref{fig:SNS} we realize that this is where the sensed nodes chosen is different than the state-averaged metric. Such change in observability degree or gain could be a result of local stability of the system dynamics associated with the perturbed initial conditions. We note, that the reason for this variability, that is postulated to be consequence of the perturbed state trajectories, is not fully investigated for that is out of the scope of this work. We here investigate the effects of the perturbations on the sensor selections and the robustness of the proposed state-averaged observability measures.

\subsection{Sensor Selections on Initial State Estimation}
The second goal is to estimate the "true" state, using the optimally selected sensors and to show that an optimal solution is obtained using a greedy heuristic approach. By showing optimality of the greedy approach, we numerically validate the modularity and  submodularity of the state-averaged observability measures.

We determine the optimal location of sensor nodes, 
then we compute the state estimate for this selection by solving $\mathbf{P1}$ with $\m Q = \m I$. State estimation results are also computed for randomly generated sensor locations under a fixed sensor fraction. The least-squares problem $\mathbf{P1}$ is solved using the MATLAB function \texttt{lsqnonlin} which implements the  trust-region-reflective algorithm. We show the results for the two metrics, $\mr{log}$-$\mr{det}\text{ and }\mr{trace}$. We quantify the final estimation performance by computing the relative estimation error using the following formula $e=\left\|\m {x}_{\text{true}}-\hat{\m {x}} \right\|_{2}/\left\|\m {x}_{\text{true}}  \right\|_{2}$, where $\m {x}_{true}$ is the true state that we want to estimate and $\hat{\m {x}}$ is its estimate computed by solving the nonlinear least squares problem for the fixed sensor location. The relative estimation error is directly related to the degree of observability resulting from the placements. Meaning that for a given sensed node configuration that achieves a high observability degree, the relative error in state estimation would be minimal. 

Fig.~\ref{fig:relativeerror} shows computed relative errors for different fraction of sensor nodes. The errors represented by the histograms are computed for a random selection of sensor location for specified sensor fractions.
At the same time we compute the relative error produced by the state-averaged approach for the two measures (red and yellow vertical lines). The $\mr{log}$-$\mr{det}\text{ and }\mr{trace}$ measures produce optimal results for each of the specified sensor fraction. This shows that the greedy algorithm indeed results in optimal sensor placement under the state-averaged approach. Thereby providing numerical proof that submodularity and modularity is retained. And on this note, we conclude this section. 
\section{Paper Summary and Future Work}\label{sec:conclusion}
This paper investigates the robustness of an observability-based sensor selection problem towards unknown initial conditions. Specifically, our approach is built upon the open-loop lifted observer framework in which the parameterized observability Gramian is constructed. To accommodate the inaccuracy when quantifying the observability due to uncertain initial states, we introduce state-averaged observability measures. The proposed sensor selection problem posed under $\mr{trace}$ and $\mr{log}$-$\mr{det}$ measures is shown to retain the modularity or submodularity properties. Greedy heuristics are employed to efficiently solve the optimization problem and render it scalable to larger nonlinear systems. Numerical results showcase the validity and effectiveness of proposed approach. For our future work, we will further investigate the relation between the proposed state-averaged observability metrics and empirical observability Gramian for discrete-time systems, and how to normalize nonlinear systems that have several basins of attraction such that the sensor placement is robust.
\bibliographystyle{IEEEtran}
\bibliography{bib_file}

\begin{thebibliography}{10}
\providecommand{\url}[1]{#1}
\csname url@samestyle\endcsname
\providecommand{\newblock}{\relax}
\providecommand{\bibinfo}[2]{#2}
\providecommand{\BIBentrySTDinterwordspacing}{\spaceskip=0pt\relax}
\providecommand{\BIBentryALTinterwordstretchfactor}{4}
\providecommand{\BIBentryALTinterwordspacing}{\spaceskip=\fontdimen2\font plus
\BIBentryALTinterwordstretchfactor\fontdimen3\font minus
  \fontdimen4\font\relax}
\providecommand{\BIBforeignlanguage}[2]{{%
\expandafter\ifx\csname l@#1\endcsname\relax
\typeout{** WARNING: IEEEtran.bst: No hyphenation pattern has been}%
\typeout{** loaded for the language `#1'. Using the pattern for}%
\typeout{** the default language instead.}%
\else
\language=\csname l@#1\endcsname
\fi
#2}}
\providecommand{\BIBdecl}{\relax}
\BIBdecl

\bibitem{Taylor2017}
J.~A. {Taylor}, N.~{Luangsomboon}, and D.~{Fooladivanda}, ``Allocating sensors
  and actuators via optimal estimation and control,'' \emph{IEEE Transactions
  on Control Systems Technology}, vol.~25, no.~3, pp. 1060--1067, 2017.

\bibitem{Berry2005}
J.~W. Berry, L.~Fleischer, W.~E. Hart, C.~A. Phillips, and J.-P. Watson,
  ``Sensor placement in municipal water networks,'' \emph{Journal of Water
  Resources Planning and Management}, vol. 131, no.~3, pp. 237--243, 2005.

\bibitem{Mehr2018}
N.~{Mehr} and R.~{Horowitz}, ``A submodular approach for optimal sensor
  placement in traffic networks,'' in \emph{2018 Annual American Control
  Conference (ACC)}, 2018, pp. 6353--6358.

\bibitem{ZHANG2017202}
H.~Zhang, R.~Ayoub, and S.~Sundaram, ``{Sensor selection for Kalman filtering
  of linear dynamical systems: Complexity, limitations and greedy
  algorithms},'' \emph{Automatica}, vol.~78, pp. 202--210, 2017.

\bibitem{Pequito2016}
S.~{Pequito}, S.~{Kar}, and A.~P. {Aguiar}, ``A framework for structural
  input/output and control configuration selection in large-scale systems,''
  \emph{IEEE Transactions on Automatic Control}, vol.~61, no.~2, pp. 303--318,
  2016.

\bibitem{Dhingra2014}
N.~K. {Dhingra}, M.~R. {Jovanović}, and Z.~{Luo}, ``{An ADMM algorithm for
  optimal sensor and actuator selection},'' in \emph{53rd IEEE Conference on
  Decision and Control}, 2014, pp. 4039--4044.

\bibitem{Argha2019}
A.~Argha, S.~W. Su, A.~Savkin, and B.~Celler, ``A framework for optimal
  actuator/sensor selection in a control system,'' \emph{International Journal
  of Control}, vol.~92, no.~2, pp. 242--260, 2019.

\bibitem{Taha2019timevarying}
A.~F. {Taha}, N.~{Gatsis}, T.~{Summers}, and S.~A. {Nugroho}, ``Time-varying
  sensor and actuator selection for uncertain cyber-physical systems,''
  \emph{IEEE Transactions on Control of Network Systems}, vol.~6, no.~2, pp.
  750--762, 2019.

\bibitem{Joshi2009}
S.~Joshi and S.~Boyd, ``{Sensor selection via convex optimization},''
  \emph{IEEE Transactions on Signal Processing}, vol.~57, no.~2, pp. 451--462,
  2009.

\bibitem{Summers2016Submodularity}
T.~H. {Summers}, F.~L. {Cortesi}, and J.~{Lygeros}, ``On submodularity and
  controllability in complex dynamical networks,'' \emph{IEEE Transactions on
  Control of Network Systems}, vol.~3, no.~1, pp. 91--101, 2016.

\bibitem{Nugroho2019}
S.~A. Nugroho, A.~F. Taha, N.~Gatsis, T.~H. Summers, and R.~Krishnan,
  ``Algorithms for joint sensor and control nodes selection in dynamic
  networks,'' \emph{Automatica}, vol. 106, pp. 124--133, Aug. 2019.

\bibitem{Shen2014sensor}
X.~{Shen}, S.~{Liu}, and P.~K. {Varshney}, ``Sensor selection for nonlinear
  systems in large sensor networks,'' \emph{IEEE Transactions on Aerospace and
  Electronic Systems}, vol.~50, no.~4, pp. 2664--2678, 2014.

\bibitem{Qi2015empirical}
J.~{Qi}, K.~{Sun}, and W.~{Kang}, ``Optimal {PMU} placement for power system
  dynamic state estimation by using empirical observability {Gramian},''
  \emph{IEEE Transactions on Power Systems}, vol.~30, no.~4, pp. 2041--2054,
  2015.

\bibitem{Haber2017}
A.~{Haber}, F.~{Molnar}, and A.~E. {Motter}, ``State observation and sensor
  selection for nonlinear networks,'' \emph{IEEE Transactions on Control of
  Network Systems}, vol.~5, no.~2, pp. 694--708, 2018.

\bibitem{Bopardikar2019}
S.~D. {Bopardikar}, O.~{Ennasr}, and X.~{Tan}, ``Randomized sensor selection
  for nonlinear systems with application to target localization,'' \emph{IEEE
  Robotics and Automation Letters}, vol.~4, no.~4, pp. 3553--3560, 2019.

\bibitem{Nugroho2019Sensor}
S.~A. {Nugroho} and A.~F. {Taha}, ``Sensor placement strategies for some
  classes of nonlinear dynamic systems via {Lyapunov} theory,'' in \emph{2019
  IEEE 58th Conference on Decision and Control (CDC)}, 2019, pp. 4551--4556.

\bibitem{Haber2021}
A.~{Haber}, S.~A. {Nugroho}, P.~{Torres}, and A.~F. {Taha}, ``Control node
  selection algorithm for nonlinear dynamic networks,'' \emph{IEEE Control
  Systems Letters}, vol.~5, no.~4, pp. 1195--1200, 2021.

\bibitem{iserles2008first}
A.~Iserles, \emph{A First Course in the Numerical Analysis of Differential
  Equations}, ser. Cambridge Texts in Applied Mathematics.\hskip 1em plus 0.5em
  minus 0.4em\relax Cambridge University Press, 2008.

\bibitem{Atkinson2011}
K.~E. Atkinson, W.~Han, and D.~Stewart, \emph{{Numerical Solution of Ordinary
  Differential Equations}}, 2011.

\bibitem{Hanba2009}
S.~{Hanba}, ``On the “uniform” observability of discrete-time nonlinear
  systems,'' \emph{IEEE Transactions on Automatic Control}, vol.~54, no.~8, pp.
  1925--1928, 2009.

\bibitem{nocedal2006numerical}
J.~Nocedal and S.~Wright, \emph{Numerical Optimization}.\hskip 1em plus 0.5em
  minus 0.4em\relax Springer Science \& Business Media, 2006.

\bibitem{Kazma2023}
M.~H. Kazma and A.~F. Taha, ``{Optimal Placement of PMUs in Power Networks:
  Modularity Meets A Priori Optimization},'' \emph{Proceedings of the American
  Control Conference}, vol. 2023-June, pp. 4489--4494, 2023.

\bibitem{Krause2008}
A.~Krause, J.~Leskovec, C.~Guestrin, J.~VanBriesen, and C.~Faloutsos,
  ``{Efficient Sensor Placement Optimization for Securing Large Water
  Distribution Networks},'' \emph{Journal of Water Resources Planning and
  Management}, vol. 134, no.~6, pp. 516--526, 2008.

\bibitem{lovasz1983submodular}
L.~Lov{\'a}sz, ``Submodular functions and convexity,'' in \emph{Mathematical
  programming the state of the art}.\hskip 1em plus 0.5em minus 0.4em\relax
  Springer, 1983, pp. 235--257.

\bibitem{nemhauser1978analysis}
G.~L. Nemhauser, L.~A. Wolsey, and M.~L. Fisher, ``An analysis of
  approximations for maximizing submodular set functions—i,''
  \emph{Mathematical programming}, vol.~14, no.~1, pp. 265--294, 1978.

\bibitem{Bach2013b}
F.~Bach, ``{Learning with submodular functions: A convex optimization
  perspective},'' \emph{Foundations and Trends in Machine Learning}, vol.~6,
  no. 2-3, pp. 145--373, 2013.

\bibitem{Krause2011}
A.~Krause and D.~Golovin, ``{Submodular function maximization},'' pp. 71--104,
  2011.

\bibitem{Zhou2019}
L.~Zhou and P.~Tokekar, ``{Sensor Assignment Algorithms to Improve
  Observability while Tracking Targets},'' \emph{IEEE Transactions on
  Robotics}, vol.~35, no.~5, pp. 1206--1219, 2019.

\bibitem{Bilmes2022}
J.~Bilmes, ``{Submodularity In Machine Learning and Artificial Intelligence},''
  2022.

\bibitem{turns1996introduction}
S.~Turns, \emph{An Introduction to Combustion: Concepts and Applications}, ser.
  McGraw-Hill series in mechanical engineering.\hskip 1em plus 0.5em minus
  0.4em\relax McGraw-Hill, 1996.

\bibitem{smirnov2014modeling}
N.~Smirnov and V.~Nikitin, ``Modeling and simulation of hydrogen combustion in
  engines,'' \emph{International Journal of Hydrogen Energy}, vol.~39, no.~2,
  pp. 1122 -- 1136, 2014.

\bibitem{Cantera}
D.~G. Goodwin, H.~K. Moffat, and R.~L. Speth, ``Cantera: An object-oriented
  software toolkit for chemical kinetics, thermodynamics, and transport
  processes,'' http://www.cantera.org, accessed: 09-27-2020.

\bibitem{Dudkowski2016}
D.~Dudkowski, S.~Jafari, T.~Kapitaniak, N.~V. Kuznetsov, G.~A. Leonov, and
  A.~Prasad, \emph{{Hidden attractors in dynamical systems}}.\hskip 1em plus
  0.5em minus 0.4em\relax Elsevier B.V., 2016, vol. 637.

\end{thebibliography}
\end{document}